\documentclass[letterpaper,10pt,conference]{IEEEtran}
\usepackage{graphicx,ifthen}
\usepackage{psfrag,amssymb,amsthm}
\usepackage[cmex10]{amsmath}
\usepackage{cite,flushend,color,pst-plot,stfloats,pst-sigsys}
\usepackage{pstricks-add}

\newtheorem{thm}{Theorem}
\newtheorem{lem}{Lemma}
\newtheorem{cor}{Corollary}
\theoremstyle{definition}
\newtheorem{definition}{Definition}

\title{Information Loss in Static Nonlinearities}

\author{\IEEEauthorblockN{Bernhard C. Geiger\IEEEauthorrefmark{1}, Christian Feldbauer\IEEEauthorrefmark{1}, Gernot Kubin\IEEEauthorrefmark{1}
\IEEEauthorblockA{\IEEEauthorrefmark{1}Signal Processing and Speech Communication Laboratory, Graz University of Technology, Austria}
\{geiger,feldbauer,gernot.kubin\}@tugraz.at}}

\begin{document}
\newcounter{myTempCnt}
\newcommand{\x}[1]{x[#1]}
\newcommand{\y}[1]{y[#1]}

\newcommand{\pdfy}{f_Y(y)}

\newcommand{\ent}[1]{H(#1)}
\newcommand{\diffent}[1]{h(#1)}
\newcommand{\derate}[1]{\bar{h}\left(#1\right)}
\newcommand{\mutinf}[1]{I(#1)}
\newcommand{\ginf}[1]{I_G(#1)}
\newcommand{\kld}[2]{D(#1||#2)}
\newcommand{\binent}[1]{H_2(#1)}
\newcommand{\binentneg}[1]{H_2^{-1}\left(#1\right)}
\newcommand{\normmi}[1]{\ent{#1}}

\newcommand{\dom}[1]{\mathcal{#1}}
\newcommand{\indset}[1]{\mathbb{I}\left({#1}\right)}

\newcommand{\unif}[2]{\mathcal{U}\left(#1,#2\right)}
\newcommand{\chis}[1]{\chi^2\left(#1\right)}
\newcommand{\chir}[1]{\chi\left(#1\right)}
\newcommand{\normdist}[2]{\mathcal{N}\left(#1,#2\right)}
\newcommand{\Prob}[1]{\mathrm{Pr}(#1)}

\newcommand{\expec}[1]{\mathrm{E}\left\{#1\right\}}
\newcommand{\expecwrt}[2]{\mathrm{E}_{#1}\left\{#2\right\}}
\newcommand{\var}[1]{\mathrm{Var}\left\{#1\right\}}
\renewcommand{\det}[1]{\mathrm{det}\left\{#1\right\}}
\newcommand{\cov}[1]{\mathrm{Cov}\left\{#1\right\}}
\newcommand{\sgn}[1]{\mathrm{sgn}\left(#1\right)}
\newcommand{\sinc}[1]{\mathrm{sinc}\left(#1\right)}
\newcommand{\e}[1]{\mathrm{e}^{#1}}
\newcommand{\multint}{\iint{\cdots}\int}

\newcommand{\hvec}{\mathbf{h}}
\newcommand{\avec}{\mathbf{a}}
\newcommand{\fvec}{\mathbf{f}}
\newcommand{\vvec}{\mathbf{v}}
\newcommand{\xvec}{\mathbf{x}}
\newcommand{\Xvec}{\mathbf{X}}
\newcommand{\Xhvec}{\hat{\mathbf{X}}}
\newcommand{\xhvec}{\hat{\mathbf{x}}}
\newcommand{\xtvec}{\tilde{\mathbf{x}}}
\newcommand{\Yvec}{\mathbf{Y}}
\newcommand{\yvec}{\mathbf{y}}
\newcommand{\Hmat}{\mathbf{H}}
\newcommand{\Amat}{\mathbf{A}}
\newcommand{\Fmat}{\mathbf{F}}

\newcommand{\zerovec}{\mathbf{0}}
\newcommand{\eye}{\mathbf{I}}
\newcommand{\evec}{\mathbf{i}}

\newcommand{\zeroone}{\left[\begin{array}{c}\zerovec^T\\ \eye\end{array} \right]}
\newcommand{\zerooneT}{\left[\begin{array}{cc}\zerovec & \eye\end{array} \right]}
\newcommand{\zerooneM}{\left[\begin{array}{cc}\zerovec &\zerovec^T\\\zerovec& \eye\end{array} \right]}

\newcommand{\Cxx}{\mathbf{C}_{XX}}
\newcommand{\Cxh}{\mathbf{C}_{\hat{X}\hat{X}}}
\newcommand{\rxx}{\mathbf{r}_{XX}}
\newcommand{\Cxy}{\mathbf{C}_{XY}}
\newcommand{\Cyy}{\mathbf{C}_{YY}}
\newcommand{\Cnn}{\mathbf{C}_{NN}}
\newcommand{\Cyx}{\mathbf{C}_{YX}}
\newcommand{\Cygx}{\mathbf{C}_{Y|X}}

\newcommand{\NN}{{N{\times}N}}
\newcommand{\perr}{P_e}
\newcommand{\perh}{\hat{\perr}}
\newcommand{\pert}{\tilde{\perr}}

\maketitle

\begin{abstract}
In this work, conditional entropy is used to quantify the information loss induced by passing a continuous random variable through a memoryless nonlinear input-output system. We derive an expression for the information loss depending on the input density and the nonlinearity and show that the result is strongly related to the non-injectivity of the considered system. Tight upper bounds are presented, which can be evaluated with less difficulty than a direct evaluation of the information loss, which involves the logarithm of a sum. Application of our results is illustrated on a set of examples.
\end{abstract}

\section{Introduction}\label{sec:intro}
Information processing, in the sense of changing the information of or retrieving information from a signal, can only be accomplished by nonlinear systems, while causal, stable, linear systems do not affect a signal's entropy rate~\cite{Shannon_TheoryOfComm},~\cite[pp.~663]{Papoulis_Probability}. As a consequence, in the past information-theoretic measures in system analysis were almost exclusively used for highly nonlinear, chaotic systems, mainly motivated by the works of Kolmogorov~\cite{Kolmogorov_Entropy1,Kolmogorov_Entropy2} and Sinai~\cite{Sinai_Entropy}. On the contrary, linear systems and relatively simple nonlinear systems (e.g., containing static nonlinearities) usually lack information-theoretic descriptions and are often characterized by second-order statistics or energetic measures (e.g., transfer function, power spectrum, signal-to-distortion ratio, mean square error between input and output, correlation functions, etc.).

In this work, we characterize the amount of information lost by passing a signal through a static nonlinearity. These systems, although simple, are by no means irrelevant in technical applications: One of the major components of the energy detector, a low-complexity receiver architecture for wireless communications, is a square-law device. Rectifiers, omnipresent in electronic systems are another example for static nonlinearities, which further constitute the nonlinear components in Wiener and Hammerstein systems. This work thus acts as a first step towards the goal of a comprehensive information-theoretic framework for more general nonlinear systems, providing an alternative to the prevailing energetic descriptions. While an analysis of information rates will be left for future work, this paper is concerned with zeroth-order entropies only.

Information loss can most generally be expressed as the difference of mutual informations,
\begin{equation}
 \mutinf{\hat{X};X}-\mutinf{\hat{X};Y}\label{eq:genloss}
\end{equation}
where the random variables (RV) $X$ and $Y$ are two descriptions for another RV $\hat{X}$.
In words, the difference in~\eqref{eq:genloss} is the information lost by changing the description from $X$ to $Y$ (cf. Fig.~\ref{fig:sysmod}). This kind of information loss is of particular interest for learning/coding/clustering (e.g., word clustering~\cite{Dhillon_LossLearning}) and triggered the development of optimal representation techniques~\cite{Tishby_InformationBottleneck}. Generally, changing the description from $X$ to $Y$ does not necessarily imply that the information loss is non-negative. In the special case that $Y$ is a function of $X$ -- the case we are concerned with -- the data processing inequality states that information can only be lost~\cite[pp.~35]{Cover_Information2}. In other words, the difference in~\eqref{eq:genloss} is non-negative.

In case $\hat{X}$ is identical to the RV $X$ itself, the information loss simplifies to (cf. proof of Theorem~\ref{thm:InfoLoss})
\begin{equation}
 \ent{X|Y}
\end{equation}
i.e., to the conditional entropy of $X$ given the description $Y$. This equivocation, as Shannon termed it in his seminal paper~\cite{Shannon_TheoryOfComm}, was originally used to describe the information loss for stochastic relations between the RVs $X$ and $Y$. In contrary to that, we are concerned with deterministic functions $Y=g(X)$. 

To our knowledge, little work has been done in this regard. Some results are available for the capacity of nonlinear channels~\cite{Zillmann_NonlinearChannel,Abou-Faycal_CapacityNLChannels}, and recently the capacity of a noisy (possibly nonlinear and non-injective) function was analyzed~\cite{Simon_NoisyFunct,Nazer_ComputationMAC}. Considering deterministic systems, we found that Pippenger used equivocation to characterize the information loss induced by multiplying two integer numbers~\cite{Pippenger_MultLoss}, while the coarse observation of discrete stochastic processes is analyzed in~\cite{Watanabe_InfoLoss}. An analysis of how much information is lost by passing a continuous RV through a static nonlinearity cannot be found in the literature.

Aside from providing information-theoretic descriptions for the nonlinear systems mentioned above, our results also apply to different fields of signal processing and communication theory. To be specific, the information loss may prove useful to compute error bounds for the reconstruction of nonlinearly distorted signals~\cite[pp.~38]{Cover_Information2} and in capacity considerations for nonlinear channels. To give another example, according to~\cite{Simon_NoisyFunct} the capacity of a noisy function $G(\cdot)$ (a noisy implementation of the determinisitic function $g(\cdot)$), is given as the maximum of
\begin{equation}
 \ent{X|Y} + \mutinf{G(X);Y}
\end{equation}
over all (discrete) distributions of $X$. In this work, we give an expression for the first part of this equation, assuming that $X$ is a continuous RV.

After introducing the problem statement in Section~\ref{sec:problem}, an expression for the information loss is derived and related to the non-injectivity of the system in Section~\ref{sec:derivation}, while bounds on the information loss are presented in Section~\ref{sec:bounds}. Section~\ref{sec:examples} illustrates the theoretical results with the help of examples.

This is an extended version of a paper submitted to an IEEE conference~\cite{Geiger_Conf2011}.

\section{Problem Statement}
\label{sec:problem}
We focus our attention on a class of systems whose input-output behavior can be described by a piecewise strictly monotone function. While this excludes functions which are constant on some proper interval (e.g., limiters or quantizers, for which it can be shown that the information loss becomes infinite), many well-behaved functions can be interpreted in the light of the forthcoming Definition. Take, e.g., the function $g(x)=\cos(x)$ for some $\dom{X}=[0,L\pi)$. While the function is clearly not monotone on $\dom{X}$, it is strictly monotone for all $\dom{X}_i=[(i-1)\pi,i\pi)$, $i=1,\dots,L$. As it turns out, piecewise strict monotonicity does not rule out functions whose derivative is zero on a finite set. In addition to that, neither continuity nor differentiability are requirements imposed by Definition~\ref{def:function}, but only piecewise continuity and piecewise differentiability.
\begin{definition}\label{def:function}
Let $g{:}\ \dom{X}\to \dom{Y}$, $\dom{X},\dom{Y}\subseteq \mathbb{R}$, be a bounded, surjective, Borel measurable function which is piecewise strictly monotone on $L$ subdomains $\dom{X}_l$
\begin{eqnarray}
 g(x) = \begin{cases}
         g_1(x), & \text{if }  x\in\dom{X}_1\\
g_2(x), &\text{if }  x\in\dom{X}_2\\
\vdots\\
g_L(x), & \text{if } x\in\dom{X}_L
        \end{cases}
\end{eqnarray}
where $g_l{:}\ \dom{X}_l\to\dom{Y}_l$ are bijective. We assume that the subdomains are an ordered set of disjoint, proper intervals  with $\bigcup_{l=1}^L \dom{X}_l =\dom{X}$ and $x_i<x_j$ for all $x_i\in\dom{X}_i$, $x_j\in\dom{X}_j$ whenever $i<j$. We further require all $g_l(\cdot)$ to be differentiable on the interval enclosure of $\dom{X}_l$.
\end{definition}
Note that $\dom{X}$ does not need to be an interval itself. Strict monotonicity implies that the function is invertible on each interval $\dom{X}_l$, i.e., there exists an inverse function $g_l^{-1}{:}\ \dom{Y}_l \to\dom{X}_l$, where $\dom{Y}_l$ is the image of $\dom{X}_l$. However, the function $g(\cdot)$ needs not be invertible on $\dom{X}$, i.e., it can be non-injective.
Equivalently, the images of the intervals, $\dom{Y}_l$, unite to $\dom{Y}$, but need not be disjoint. Let $g(\cdot)$ describe the input-output behavior of the system under consideration (see Fig.~\ref{fig:sysmod}).

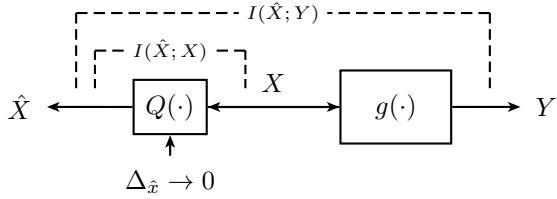
\begin{figure}[t]
\centering
\begin{pspicture}[showgrid=false](1,1)(8,3.5)
 	\pssignal(1,2){x}{$\hat{X}$}
 	\pssignal(3,1){n}{$\Delta_{\hat{x}} \to 0$}
	\psfblock[framesize=1 0.75](3,2){oplus}{$Q(\cdot)$}
	\psfblock[framesize=1.5 1](6,2){c}{$g(\cdot)$}
	\pssignal(8,2){y}{$Y$}
  \ncline[style=Arrow]{n}{oplus}
	\ncline[style=Arrow]{oplus}{x}
 \nclist[style=Arrow]{ncline}[naput]{oplus,c $X$,y}
  \ncline[style=Arrow]{c}{oplus}
	\psline[style=Dash](2,2.75)(4,2.75)
	\psline[style=Dash](2,2.25)(2,2.75)
	\psline[style=Dash](4,2.25)(4,2.75)
	\psline[style=Dash](1.75,3.25)(7.25,3.25)
	\psline[style=Dash](1.75,2.25)(1.75,3.25)
	\psline[style=Dash](7.25,2.25)(7.25,3.25)
 	\rput*(3,2.75){\scriptsize{$\mutinf{\hat{X};X}$}}
	\rput*(4.5,3.25){\scriptsize{$\mutinf{\hat{X};Y}$}}
\end{pspicture}
\caption{Equivalent model for computing the information loss of an input-output system with static nonlinearity $g(\cdot)$. $Q(\cdot)$ is a quantizer with quantization step size $\Delta_{\hat{x}}$. Note that $X$ can be modeled as the sum of $\hat{X}$ and an input-dependent finite-support noise term $N$ as in~\cite{Geiger_Conf2011}.}
 \label{fig:sysmod}
\end{figure}

As an input to this system consider a sequence of independent samples, identically distributed with continuous cumulative distribution function (CDF) $F_X(x)$ and probability density function (PDF) $f_X(x)$. Without loss of generality, let the support of this RV be $\dom{X}$, i.e., $f_X(x)$ is positive on $\dom{X}$ and zero elsewhere.

As an immediate consequence of this system model, the conditional PDF of the output $Y$ given the input $X$ can be written as~\cite{Chi_TransformingDirac}
\begin{equation}
 f_{Y|X}(x,y) = \delta(y-g(x))= \sum_{i\in\indset{y}} \frac{\delta(x-x_i)}{\left|g'\left(x_i\right)\right|} \label{eq:condPDFY}
\end{equation}
where $\delta(\cdot)$ is Dirac's delta distribution, $\indset{y} = \{i: y\in\dom{Y}_i\}$ and $x_i=g_i^{-1}(y)$ for all $i\in\indset{y}$. In other words, $\{x_i\}$ is the preimage of $y$ or the set of roots satisfying $y=g(x)$.
The marginal PDF of $Y$ is thus given as~\cite[pp.~130]{Papoulis_Probability},~\cite{Chi_TransformingDirac}
\begin{IEEEeqnarray}{RCL}
 f_Y(y) &=& \sum_{i\in\indset{y}} \frac{f_{X}(x_i)}{\left|g'\left(x_i\right)\right|}\label{eq:fy}.
\end{IEEEeqnarray}

\section{Information Loss of Static Nonlinearities}
\label{sec:derivation}
In what follows we quantify the information loss induced by $g(\cdot)$, and we show that this information loss is identical to the remaining uncertainty from which interval $\dom{X}_l$ the input $x$ originated after observing the output $y$.
The main contribution of this work is thus concentrated in the following two Theorems.

\begin{lem}\label{lem:XforY}
 Let $g{:}\ \dom{X}\to\dom{Y}$ and $f{:}\ \dom{Y}\to\dom{Z}$ be measurable functions. Further, let $X$ be a continuous RV on $\dom{X}$ and $Y=g(X)$. Then, 
\begin{equation}
\expec{f(y)} = \int_\dom{Y} f(y) f_Y(y) dy = \int_\dom{X} f(g(x)) f_X(x) dx.
\end{equation}
\end{lem}

\begin{IEEEproof}
The proof is based on the fact that for a measurable $g(\cdot)$~\cite[pp.~142, Theorem~5-1]{Papoulis_Probability}
\begin{equation}
 \expec{g(x)} = \int_\dom{X} g(x) f_X(x) dx.\label{eq:expec}
\end{equation}
Lemma~\ref{lem:XforY} follows from~\eqref{eq:expec} since for measurable $g(\cdot)$ and $f(\cdot)$ the composition $(f\circ g)(\cdot)=f(g(\cdot))$ is also measurable.
\end{IEEEproof}

\begin{thm}\label{thm:InfoLoss}
 The information loss induced by a function $g(\cdot)$ satisfying Definition~\ref{def:function} is given as
\begin{equation}
\ent{X|Y}= \int_{\dom{X}} f_X(x) \log \left( \frac{\sum_{i\in\indset{g(x)}} \frac{f_{X}(x_i)}{\left|g'\left(x_i\right)\right|}}{\frac{f_X(x) }{\left|g'\left(x\right)\right|} } \right)dx.\label{eq:informationloss}
\end{equation}
\end{thm}

\begin{IEEEproof}
Using identities from~\cite{Cover_Information2} and the model in Fig.~\ref{fig:sysmod} the conditional entropy $\ent{X|Y}$ can be calculated as
\begin{IEEEeqnarray}{RCL}
 \ent{X|Y}&=& \lim_{\hat{X}\rightarrow X}\left(\ent{\hat{X}|Y}-\ent{\hat{X}|X}\right)\notag\\
 &=& \lim_{\hat{X}\rightarrow X}\left(\ent{\hat{X}}-\ent{\hat{X}|X}-\ent{\hat{X}}+\ent{\hat{X}|Y}\right)\notag\\ 
&=& \lim_{\hat{X}\rightarrow X} \left(\mutinf{\hat{X};X} - \mutinf{\hat{X};Y}\right).\label{eq:limsub}
\end{IEEEeqnarray}
where $\hat{X}$ is a discrete RV converging surely to $X$. This auxillary RV is necessary to ensure that the (discrete) entropies we use are well-defined.
Here, motivated by the data processing inequality~\cite[pp.~34]{Cover_Information2}, we have related the conditional entropy to a difference of mutual informations, which we have introduced as the most general notion of information loss in Section~\ref{sec:intro}. In addition to that, the mutual information has the benefit that it is defined for general joint distributions~\cite[pp.~252]{Cover_Information2}, which eliminates the requirement for a discrete $\hat{X}$.

For the mutual information between $X$ and $\hat{X}$ we can write with~\cite[pp.~251]{Cover_Information2}
\begin{equation}
 \mutinf{\hat{X},X} = \int_\dom{X}\int_{\dom{X}} f_{\hat{X}X}(\hat{x},x) \log\left(\frac{f_{X|\hat{X}}(\hat{x},x)}{f_X(x)}\right)dxd\hat{x}. \label{eq:mutinf3}
\end{equation}
Similarily, with Lemma~\ref{lem:XforY} (the logarithm and all PDFs are measurable) we get for $\mutinf{\hat{X},Y}$
\begin{IEEEeqnarray}{RCL}
 \mutinf{\hat{X};Y} &=& \int_\dom{X} \int_\dom{X} f_{\hat{X}X}(\hat{x},x) \log \left(\frac{f_{Y|\hat{X}}(\hat{x},g(x))}{f_Y(g(x))}\right)  dxd\hat{x}.\notag\\\label{eq:mutinf1}
\end{IEEEeqnarray}
After subtracting these expressions according to~\eqref{eq:limsub} we can exchange limit and integration (see Appendix). In the limit the conditional PDFs assume $f_{X|\hat{X}}(\hat{x},x) = \delta(x-\hat{x})$ and $f_{Y|\hat{X}}(\cdot,\cdot)=f_{Y|X}(\cdot,\cdot)$, thus~\eqref{eq:condPDFY}, and using these we obtain~\eqref{eq:mutinf2} at the bottom of the next page.
\begin{figure*}[!b]
 \normalsize
\hrulefill
\begin{IEEEeqnarray}{RCL}
 \ent{X|Y}=\lim_{\hat{X}\rightarrow X} \left(\mutinf{\hat{X};X} - \mutinf{\hat{X};Y}\right)
&=&  \int_\dom{X}\int_{\dom{X}} f_X(x) \delta(x-\hat{x}) \log \left( 
\frac{\delta(x-\hat{x}) f_Y(g(x))}{f_X(x)\sum_{k\in\indset{g(x)}} \frac{\delta(\hat{x}-x_k)}{\left|g'\left(x_k\right)\right|}} \right)d\hat{x}dx \label{eq:mutinf2}
\end{IEEEeqnarray}
\end{figure*}
Since the integral over $\hat{x}$ is zero for $\hat{x}\neq x$ due to $\delta(x-\hat{x})$, only the term satisfying $x_k=x$ remains from the sum over Dirac's deltas in the denominator; this term cancels with the delta in the numerator. Integrating over $\hat{x}$ and substituting~\eqref{eq:fy} for $f_Y(\cdot)$ finally yields
\begin{equation}
\ent{X|Y}= \int_{\dom{X}} f_X(x) \log \left( \frac{\sum_{i\in\indset{g(x)}} \frac{f_{X}(x_i)}{\left|g'\left(x_i\right)\right|}}{\frac{f_X(x) }{\left|g'\left(x\right)\right|} } \right)dx\label{eq:equivoc}
\end{equation}
and completes the proof.

\end{IEEEproof}
Note that for $\hat{X}\to X$ both $\mutinf{\hat{X};X}$ and $\mutinf{\hat{X};Y}$ diverge to infinity, but their difference not necessarily does. Further, if for all $y=g(x)$ the preimage is a singleton ($|\indset{g(x)}|=1$ for all $x\in\dom{X}$), $g(\cdot)$ is injective (thus bijective by Definition~\ref{def:function}) and the information loss $\ent{X|Y}=0$.

In Theorem~\ref{thm:InfoLoss} we provided an expression to calculate the information loss induced by a static nonlinearity. The following Theorem is of a different nature: It gives an explanation of \emph{why} information is lost at all. Considering non-injective functions $g(\cdot)$ satisfying Definition~\ref{def:function}, multiple input values may lead to the same output -- the preimage of $y$ may contain multiple elements. Given the output, the input is uncertain only w.r.t. which of these elements has been fed into the system under consideration. Due to the piecewise strict monotonicity of $g(\cdot)$ each subdomain contains at most one element of the preimage of $y$. Therefore, the information loss is identical to the uncertainty about the interval $\dom{X}_l$ from which the input $x$ originated given the output value $y$. Before making this statement precise in Theorem~\ref{thm:equivToRoots}, let us introduce the following Definition:
\begin{definition}
 Let $W$ be a discrete RV with $|\dom{W}|=L$ mass points which is defined as
\begin{equation}
 W=w_i \ \ \text{if} \ \ x\in\dom{X}_i \label{eq:defW}
\end{equation}
for all $i=1,\dots,L$.
\end{definition}
In other words, $W$ is a discrete RV which depends on the interval $\dom{X}_l$ of $x$, and not on its actual value.
As an immediate consequence of this Definition we obtain
\begin{equation}
 \Prob{W=w_i} = p(w_i) = \int_{\dom{X}_i} f_X(x) dx
\end{equation}
i.e., the probability mass contained in the $i$-th interval.

In accordance with the model in Fig.~\ref{fig:sysmod} and the reasoning in the Appendix, one can think of $W$ as $\hat{X}$ when the quantization bins are identical to $\dom{X}_l$. While in Theorem~\ref{thm:InfoLoss} we required $\hat{X}$ to converge to $X$ surely, the next Theorem shows that indeed such a convergence is not necessary as long as the quantization bins are chosen appropriately. This fact will then establish the link between the non-injectivity, piecewise strict monotonicity on intervals, and information loss.
We are now ready to state the main Theorem:

\begin{thm}[Main Theorem]
\label{thm:equivToRoots}
 The uncertainty about the input value $x$ after observing the output $y$ is identical to the uncertainty about the interval $\dom{X}_l$ from which the input was taken, i.e.,
\begin{equation}
 \ent{X|Y}=\ent{W|Y}.
\end{equation}
\end{thm}

\begin{IEEEproof}
Permitting continuous observations $Y$ of a discrete random variable $W$, i.e., $\dom{Y}\subseteq\mathbb{R}$, we can write for the conditional entropy~\cite{Feder_EntropyError,Chu_InequalitiesErrorProb}
\begin{IEEEeqnarray}{RCL}
 \ent{W|Y} &=& \int_\dom{Y} \ent{W|Y=y} f_Y(y)dy\\
&=& -\int_\dom{Y} \sum_{i=1}^{L} p(w_i|y)\log p(w_i|y) f_Y(y)dy .\label{eq:condMixed}
\end{IEEEeqnarray}

The conditional probability of $W$ given $Y$, $p(w_i|y)=\Prob{W=w_i|Y=y}$, can be calculated from~\eqref{eq:defW} as
\begin{IEEEeqnarray}{RCL}
 p(w_i|y) &=& \int_{\dom{X}_i} f_{X|Y}(x,y) dx\\
&=& \frac{1}{f_Y(y) } \int_{\dom{X}_i} f_{Y|X}(x,y)f_X(x) dx\\
 &=& \frac{1}{f_Y(y) }\sum_{k\in\indset{y}} \int_{\dom{X}_i}  \frac{\delta(x-x_k)}{|g'(x_k)|}f_X(x) dx
\end{IEEEeqnarray}
where we made use of~\eqref{eq:condPDFY} and exchanged the order of summation and integration. Due to piecewise strict monotonicity of the function $g(\cdot)$, each interval contains at most one element of the preimage of $y$. Thus, if $i\in\indset{y}$, this element is given by $g_i^{-1}(y)$ and
\begin{IEEEeqnarray}{RCL}
 p(w_i|y) &=& \frac{1}{f_Y(y) } \int_{\dom{X}_i} \frac{\delta(x-g_i^{-1}(y))}{|g'(g_i^{-1}(y))|}f_X(x) dx\\
 &=& \frac{f_X(g_i^{-1}(y))}{|g'(g_i^{-1}(y))|f_Y(y) } .\label{eq:condpwy}
\end{IEEEeqnarray}
Conversely, if $i\notin\indset{y}$, we obtain $p(w_i|y)=0$.

We are now ready to compute the conditional entropy of $W$ given $Y$, i.e., the uncertainty about the interval from which $x$ was drawn:
\begin{IEEEeqnarray}{RCL}
 \ent{W|Y} &=& -\int_\dom{Y} \sum_{i=1}^{L} p(w_i|y) \log p(w_i|y)f_Y(y) dy\\
 &=& -\int_\dom{Y} \sum_{i\in\indset{y}} \frac{f_X(x_i)}{|g'(x_i)|} \log\left(\frac{f_X(x_i)}{|g'(x_i)|f_Y(y) }\right) dy\notag \\
\end{IEEEeqnarray}
where we used $x_i=g_i^{-1}(y)$. We can now substitute $x=g_l^{-1}(y)$ for all $l=1,\dots,L$ where in each $\dom{X}_l$ only a single root remains. With $g_l(x)=g(x)$ on $\dom{X}_l$ we thus obtain
\begin{IEEEeqnarray}{RCL}
  \ent{W|Y} &=& -\sum_l \int_{\dom{X}_l} \frac{f_X(x)}{|g'(x)|}|g'(x)| \log\left(\frac{\frac{f_X(x)}{|g'(x)|}}{f_Y(g(x)) }\right) dx\notag\\
 &=& \int_{\dom{X}} f_X(x) \log\left(\frac{ \sum_{k\in\indset{g(x)}} \frac{f_X(x_k)}{|g'(x_k)|}}{\frac{f_X(x)}{|g'(x)|}}\right) dx\notag \\
\end{IEEEeqnarray}
where we replaced $f_Y(g(x))$ by~\eqref{eq:fy} and collapsed the sum of integrals over $\dom{X}_l$ to a single integral over $\dom{X}$. Since this result is identical to~\eqref{eq:informationloss}, i.e.,
\begin{equation}
 \ent{W|Y}=\ent{X|Y}
\end{equation}
the proof is complete.
\end{IEEEproof}

The information loss induced by a function satisfying Definition~\ref{def:function} is thus only related to the roots of the equation $y=g(x)$. Conversely, if the interval $\dom{X}_l$ of $x$ is known, the exact value of $x$ can be reconstructed after observing $y$: \begin{IEEEeqnarray}{RCL}
\ent{X|Y} &=& \ent{X,W|Y}\\
&=& \ent{X|W,Y}+\ent{W|Y}\\
&=&\ent{X|W,Y}+\ent{X|Y}
\end{IEEEeqnarray}
and thus $\ent{X|W,Y}=0$.

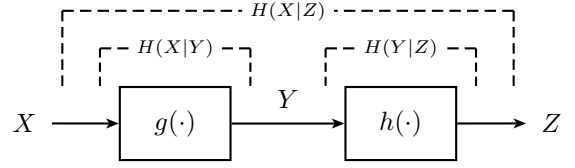
\begin{figure}[t]
 \centering
\begin{pspicture}[showgrid=false](1,1.5)(8,4)
 	\pssignal(1,2){x}{$X$}
	\psfblock[framesize=1.5 1](3,2){d}{$g(\cdot)$}
	\psfblock[framesize=1.5 1](6,2){c}{$h(\cdot)$}
	\pssignal(8,2){y}{$Z$}
  \nclist[style=Arrow]{ncline}[naput]{x,d,c $Y$,y}
	\psline[style=Dash](2,3)(4,3)
	\psline[style=Dash](2,2.5)(2,3)
	\psline[style=Dash](4,2.5)(4,3)
	\psline[style=Dash](5,3)(7,3)
	\psline[style=Dash](5,2.5)(5,3)
	\psline[style=Dash](7,2.5)(7,3)
	\psline[style=Dash](1.5,3.5)(7.5,3.5)
	\psline[style=Dash](1.5,2.5)(1.5,3.5)
	\psline[style=Dash](7.5,2.5)(7.5,3.5)
 	\rput*(3,3){\scriptsize\textcolor{black}{$\ent{X|Y}$}}
	\rput*(6,3){\scriptsize\textcolor{black}{$\ent{Y|Z}$}}
	\rput*(4.5,3.5){\scriptsize\textcolor{black}{$\ent{X|Z}$}}
\end{pspicture}
 \caption{Cascade of systems}
 \label{fig:cascade}
\end{figure}

Aside from the properties of conditional entropies (non-negativity~\cite[pp.~15]{Cover_Information2}, asymmetry in its arguments, etc.) the information loss has an important property concerning the cascade of deterministic, static systems, which is not shared by the conditional entropy in general. For such a cascade (see Fig.~\ref{fig:cascade}), which in the static case is equivalent to a composition of the implemented functions, we can prove the following Theorem:

\begin{thm}\label{thm:transitivity}
 Consider two functions $g{:}\ \dom{X}\to\dom{Y}$ and $h{:}\ \dom{Y}\to\dom{Z}$ satisfying Definition~\ref{def:function} and a cascade of systems implementing these functions, as shown in Fig.~\ref{fig:cascade}. Let $Y=g(X)$ and $Z=h(Y)$. Then, the information loss induced by this cascade, or equivalently, by the composition $(h\circ g)(\cdot)=h(g(\cdot))$ is given by:
\begin{equation}
\ent{X|Z}=\ent{X|Y}+\ent{Y|Z}
\end{equation}
\end{thm}

\begin{IEEEproof}
The proof starts by expanding $\ent{X,Y|Z}$
 \begin{IEEEeqnarray*}{RCL}
 	\ent{X,Y|Z} &=& \ent{X|Y,Z} + \ent{Y|Z}\\
&=& \ent{X|Y} + \ent{Y|Z}
\end{IEEEeqnarray*}
since $X \to Y \to Z$ forms a Markov chain and thus $X$ and $Z$ are independent given $Y$~\cite{Cover_Information2}. Further, $\ent{X,Y|Z}=\ent{X|Z}$ since $Y$ is a function of $X$. Thus,
\begin{equation}
\ent{X|Z}=\ent{X|Y}+\ent{Y|Z}
\end{equation}
and the proof is complete.
\end{IEEEproof}

Extending Theorem~\ref{thm:transitivity}, we obtain the following Corollary:
\begin{cor}
Consider a set of functions $g_i{:}\ \dom{X}_{i-1}\to\dom{X}_{i}$, $i=1,\dots,N$, each satisfying Definition~\ref{def:function}, and a cascade of systems implementing these functions. Let $X_i$, $i=1,2,\dots,N$, denote the output of the $i$th constituent system and, thus, the input of the $(i+1)$th system. Given the input of the first system, $X_0$, we have
\begin{equation}
 \ent{X_0|X_N} = \sum_{i=1}^{N} \ent{X_{i-1}|X_i}.
\end{equation}
\end{cor}

\begin{proof}
 The Corollary is proved by repeated application of Theorem~\ref{thm:transitivity}.
\end{proof}

This result does not imply that the order in which the functions can be arranged has no influence on the information loss of the cascade, as one would expect from stable, linear systems. Illustrative examples showing that this does not hold can be found, e.g., in~\cite{Johnson_ITNeural}. Moreover, calculating the individual information losses requires in each case the PDF of the input to the function under consideration. While this does not seem to yield an improvement compared to a direct evaluation of~\eqref{eq:informationloss}, Theorem~\ref{thm:transitivity} can be used to bound the information loss of the cascade efficiently whenever bounds on the individual information losses are available. We will introduce such bounds in the next Section.

\section{Upper Bounds on the Information Loss}\label{sec:bounds}
In many situations it might be inconvenient, or even impossible, to evaluate the information loss~\eqref{eq:informationloss} analytically since it involves the logarithm of a sum, for which only inequalities exist~\cite{Cover_Information2}. Therefore, one has to resort to numeric integration or use bounds on the information loss which are simpler to evaluate. In this Section we derive an upper bound which requires only minor knowledge about the function $g(\cdot)$ -- namely, the number of intervals $L$ -- and we show that this bound is tight.

\begin{figure*}[!ht]
 \centering
	\begin{pspicture}[showgrid=false](-6.0,-1)(6.0,4)
		\psaxeslabels{->}(-4,0)(-4.2,-0.2)(4,4){$x$}{$F_X(x)$, \textcolor{red}{$g(x)$}, \textcolor{blue}{$h(x)$}}
		\psplot[style=Graph]{-3}{3}{x 60 mul sin x add 3 div 1 add} \psplot[style=Graph]{-4}{-3}{0} \psplot[style=Graph]{3}{3.5}{2}
		\psplot[style=Graph,linecolor=red]{-3}{-0.5}{x 60 mul sin x add -3 div 0.33 sub}
    \psplot[style=Graph,linecolor=red]{-0.5}{0.5}{x 60 mul sin x add 3 div 0.33 add}
		\psplot[style=Graph,linecolor=red]{0.5}{3}{x 60 mul sin x add 3 div 0.33 sub}\psplot[style=Graph,linecolor=red]{-4}{-3}{0.67} \psplot[style=Graph,linecolor=red]{3}{3.5}{0.67} 
		\psplot[style=Graph,linecolor=blue]{-3}{-0.5}{x 60 mul sin x add 3 div 3 add}
    \psplot[style=Graph,linecolor=blue]{-0.5}{0.5}{x 60 mul sin x add 3 div 2.5 add}
		\psplot[style=Graph,linecolor=blue]{0.5}{3}{x 60 mul sin x add -3 div 3.5 add}\psplot[style=Graph,linecolor=blue]{-4}{-3}{2} \psplot[style=Graph,linecolor=blue]{3}{3.5}{2.5}
		\psTick{0}(-4,2) \rput[tr](-4.2,2){$1$}
		\psset{braceWidthInner=5pt,braceWidthOuter=5pt,braceWidth=0.1pt}
		\psbrace[ref=rC,linewidth=0.1pt,rot=180](-4,0.67)(-4,0){\textcolor{red}{$\dom{Z}$}}
		\psbrace[ref=rC,rot=180,linewidth=0.1pt](-4,2.67)(-4,2){\textcolor{blue}{$\dom{Y}_1$}}
		\psbrace[ref=lC,linewidth=0.1pt](0.5,2.17)(0.5,2.83){\textcolor{blue}{$\dom{Y}_2$}}
		\psbrace[ref=lC,linewidth=0.1pt](3.5,2.5)(3.5,3.17){\textcolor{blue}{$\dom{Y}_3$}}
		\psbrace[ref=t,rot=90,linewidth=0.1pt](-4,0)(-0.5,0){$\dom{X}_1$}
		\psbrace[ref=t,rot=90,linewidth=0.1pt](-0.5,0)(0.5,0){$\dom{X}_2$}
		\psbrace[ref=t,rot=90,linewidth=0.1pt](0.5,0)(3.5,0){$\dom{X}_3$}
		\psline[style=Dash,linewidth=0.01, linecolor=gray](0.5,0)(0.5,3.5)
		\psline[style=Dash,linewidth=0.01, linecolor=gray](-0.5,0)(-0.5,3.5)
		\psline[style=Dash,linewidth=0.01](-0.5,2.17)(0.5,2.17) \psline[style=Dash,linewidth=0.01](0.5,3.17)(3.5,3.17)
		\psline[style=Dash,linewidth=0.01](-4,2.67)(-0.5,2.67)
		\psset{linecolor=blue}
		\pszero(-0.5,2.67){p1} \dotnode(-0.5,2.17){d} \pszero(0.5,2.83){p1} \dotnode(0.5,3.17){d}
		\psset{linecolor=red}
		\pszero(0.5,0.67){p1} \dotnode(0.5,0){p1}
	\end{pspicture}
\caption{Piecewise strictly monotone functions with $L=3$ satisfying conditions of Theorem~\ref{thm:UpperBoundLoss}. The function in blue, $h{:}\ \dom{X}\to\dom{Y}$, renders~\eqref{eq:reqForBound} piecewise constant but not constant due to improper setting of the constants $c_l$. Tightness is achieved in the smallest bound,~\eqref{eq:bound1}, only. The function in red, $g{:}\ \dom{X}\to\dom{Z}$, satisfies all conditions (i.e.,~\eqref{eq:reqForBound} is constant and $\dom{Z}_l=\dom{Z}$ for all $l$) and thus achieves equality in the largest bound~\eqref{eq:bound3}. Note that the subdomains $\dom{X}_l$ are chosen such that each subdomain contains the same probability mass.}
\label{fig:satTh4}
\end{figure*}
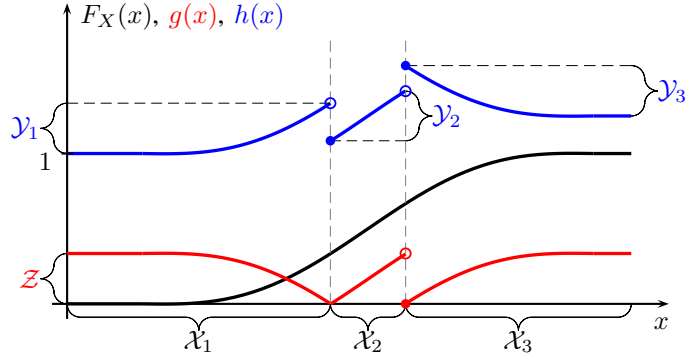

\begin{thm}\label{thm:UpperBoundLoss}
 The information loss induced by a function $g(\cdot)$ satisfying Definition~\ref{def:function} can be upper bounded by the following ordered set of inequalities:
\begin{IEEEeqnarray}{RCL}
 \ent{X|Y} &\leq& \int_\dom{Y} f_Y(y) \log \left(|\indset{y}|\right) dy \label{eq:bound1}\\
 &\leq&  \log \left(\sum_{l=1}^L \int_{\dom{Y}_l} f_Y(y) dy \right)\label{eq:bound2} \\&\leq& \log L\label{eq:bound3}
\end{IEEEeqnarray}
Bound~\eqref{eq:bound1} holds with equality if and only if
\begin{equation}
 \sum_{k\in\indset{g(x)}} \frac{f_X(x_k)}{|g'(x_k)|}\frac{|g'(x)|}{f_X(x)} \label{eq:reqForBound}
\end{equation}
is piecewise constant. If this expression is constant for all $x\in\dom{X}$, bound~\eqref{eq:bound2} is tight. Bound~\eqref{eq:bound3} holds with equality if and only if additionally $\dom{Y}_l=\dom{Y}$ for all $l=1,\dots,L$, and thus~\eqref{eq:reqForBound} evaluates to $L$.
\end{thm}

\begin{IEEEproof}The proof relies upon Theorem~\ref{thm:equivToRoots}, where we established $\ent{X|Y}=\ent{W|Y}$ and thus
\begin{equation}
\ent{X|Y} =\int_\dom{Y} \ent{W|Y=y} f_Y(y)dy.
\end{equation}
The first inequality, based upon the maximum entropy property of the uniform distribution~\cite[pp.~29]{Cover_Information2}, becomes an equality if $p(w_i|y)=\frac{1}{|\indset{y}|}$ for all $i\in\indset{y}$, such that $\ent{W|Y=y}=\log\left(|\indset{y}|\right)$. Combining this with~\eqref{eq:condpwy} we have
\begin{IEEEeqnarray}{RCL}
 p(w_i|y) &=& \frac{f_X(g_i^{-1}(y))}{|g'(g_i^{-1}(y))|f_Y(y)} = \frac{1}{|\indset{y}|}
\end{IEEEeqnarray}
Performing multiplicative inversion and inserting~\eqref{eq:fy} we obtain
\begin{IEEEeqnarray}{RCL}
|\indset{y}| &=& f_Y(y)\frac{|g'(g_i^{-1}(y))|}{f_X(g_i^{-1}(y))}\\
  &=&\sum_{k\in\indset{y}} \frac{f_X(x_k)}{|g'(x_k)|}\frac{|g'(g_i^{-1}(y))|}{f_X(g_i^{-1}(y))}
\end{IEEEeqnarray}
for all $y\in\dom{Y}$ and all $i\in\indset{y}\subseteq\{1,\dots,L\}$. Since generally $|\indset{y}|$ is piecewise constant and independent of $i$ as long as $i\in\indset{y}$, it is immaterial which $i$ from $\indset{y}$ is chosen. Thus, one can exchange $g_i^{-1}(y)$ by $x$ and $\indset{y}$ by $\indset{g(x)}$, which proves the requirement of the first bound.

The second inequality is due to Jensen~\cite[pp.~27]{Cover_Information2}. Equality is achieved if $|\indset{y}|$ is constant for all $y\in\dom{Y}$, or equivalently
\begin{IEEEeqnarray}{RCL}
  \sum_{k\in\indset{g(x)}} \frac{f_X(x_k)}{|g'(x_k)|}\frac{|g'(x)|}{f_X(x)} &=& \text{const.}
\end{IEEEeqnarray}
for all $x\in\dom{X}$.

If the requirements for equality in~\eqref{eq:bound1} and~\eqref{eq:bound2} are met, the information loss is given as $\ent{X|Y}=\log (|\indset{y}|)$ for any $y\in\dom{Y}$. Thus, the last bound,~\eqref{eq:bound3}, is tight if and only if $|\indset{y}|=L$ for all $y\in\dom{Y}$. This requires that $\dom{Y}_l=\dom{Y}$ for all $l=1,\dots,L$ and completes the proof.
\end{IEEEproof}

An example of a function $g(\cdot)$ for which~\eqref{eq:reqForBound} is piecewise constant assumes on each interval $\dom{X}_l$ the cumulative distribution function $F_X(x)$, possibly modified with a sign and an additive constant. In other words, for all $l=1,\dots,L$
\begin{equation}
 g_l(x) = b_lF_X(x) + c_l
\end{equation}
where $b_l\in\{1,-1\}$ and $c_l\in\mathbb{R}$ are arbitrary constants. Such a function $h{:}\ \dom{X}\to\dom{Y}$ is depicted in Fig.~\ref{fig:satTh4}. The constants $c_l$ and the probability masses in each interval are constrained if~\eqref{eq:reqForBound} shall be constant. As a special case, consider this constant to be equal to $L$, which guarantees tightness in the largest bound~\eqref{eq:bound3}. In order that appropriate constants $c_l$ exist, all intervals $\dom{X}_l$ have to contain the same probability mass, i.e.,
\begin{equation}
 \int_{\dom{X}_l} f_X(x) dx = \frac{1}{L}.
\end{equation}
Since equal probability mass in all intervals is a necessary, but not a sufficient condition for equality in~\eqref{eq:bound3} (cf.~Fig.~\ref{fig:satTh4}), the constants $c_l$ have to be set to
\begin{equation}
 c_l = -\sum_{i=1}^{l-1} \int_{\dom{X}_i} f_X(x) dx = -\frac{l-1}{L}\label{eq:cls}
\end{equation}
where we assume that the intervals are ordered and where $b_l=1$ for all $l$. A function $g{:}\ \dom{X}\to\dom{Z}$ satisfying these requirements is shown in Fig.~\ref{fig:satTh4}.

Another example of a function satisfying the tightness conditions of Theorem~\ref{thm:UpperBoundLoss} is given in Example 1 of Section~\ref{sec:examples}.

\section{Examples}
\label{sec:examples}
In this Section, the application of the obtained expression for the information loss and its upper bounds is illustrated. Unless otherwise noted, the logarithm is taken to base $2$.

\subsection{Example 1: Even PDF, Magnitude Function}\label{ssec:ex1}
Consider a continuous RV $X$ with an even PDF, i.e., $f_X(-x)=f_X(x)$. Let the support $\dom{X}=\mathbb{R}$ and let this RV be the input to the magnitude function, i.e.,
\begin{equation}
 g(x)=|x|=
\begin{cases}
 -x,& \text{if }x<0\\x,&\text{if }x\geq0
\end{cases}
.\label{eq:magnitude}
\end{equation}
The magnitude function is piecewise strictly monotone on $\dom{X}_1=(-\infty,0)$ and $\dom{X}_2=[0,\infty)$, and with $L=2$ we obtain the largest bound from Theorem~\ref{thm:UpperBoundLoss} as
\begin{equation}
 \ent{X|Y} \leq \log 2=1.
\end{equation}
Both intervals are mapped to the positive (non-negative) real axis, i.e., $\dom{Y}_1\cup \{0\}=\dom{Y}_2=\dom{Y}=[0,\infty)$, which implies that the second bound in Theorem~\ref{thm:UpperBoundLoss} also yields $\ent{X|Y}\leq1$. The magnitude of the first derivative of $g(\cdot)$ is equal to unity on both $\dom{X}_1$ and $\dom{X}_2$. There are two partial inverses mapping $\dom{Y}$ to the subdomains of $\dom{X}$:
\begin{IEEEeqnarray}{RCL}
 x_1=g_1^{-1}(y)&=&-y=-g(x)\text{, and}\\
 x_2=g_2^{-1}(y)&=&y=g(x).
\end{IEEEeqnarray}
Thus for all $x\in\dom{X}$ we have $|\indset{g(x)}|=2$, which renders the smallest bound of Theorem~\ref{thm:UpperBoundLoss} as $\ent{X|Y}\leq 1$. 
Combining~\eqref{eq:magnitude} with the two partial inverses, we obtain for $x\in\dom{X}_1$:
\begin{IEEEeqnarray}{RCL}
 x_1&=&x\text{, and}\\
 x_2&=&-x.
\end{IEEEeqnarray}
Conversely, for $x\in\dom{X}_2$ we have $x_1=-x$ and $x_2=x$. Using this in~\eqref{eq:informationloss} we obtain the information loss
\begin{IEEEeqnarray}{RCL}
\ent{X|Y} &=& \int_{\dom{X}_1} f_X(x) \log \left( \frac{f_X(x)+f_X(-x) }{f_X(x) } \right)dx\notag\\
 &&{}\int_{\dom{X}_2} f_X(x) \log \left( \frac{f_X(-x)+f_X(x) }{f_X(x) } \right)dx\notag\\
&=& \log 2 \int_\dom{X} f_X(x) dx=1\notag
\end{IEEEeqnarray}
which shows that all bounds of Theorem~\ref{thm:UpperBoundLoss} are tight in this example.

The conditional entropy is identical to one bit. In other words, if an RV with an even PDF (thus, with equal probability masses for positive and negative values) is fed through a magnitude function, one bit of information is lost. Despite the fact that this result seems obvious, this is the first time that it is derived for a continuous input to the best knowledge of the authors.

\subsection{Example 2: Zero-Mean Uniform PDF, Piecewise Strictly Monotone Function}
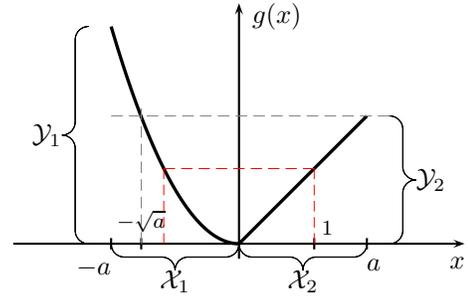
\begin{figure}[t]
 \centering
	\begin{pspicture}[showgrid=false](-4.0,-0.5)(4.0,3.2)
		\psaxeslabels{->}(0,0)(-3,-0.2)(3,3.2){$x$}{$g(x)$}
		\psplot[style=Graph]{0}{1.7}{x} \psplot[style=Graph]{-1.7}{0}{x x mul}
		\psTick{90}(-1.7,0) \rput[tr](-1.7,-0.2){$-a$}
		\psTick{90}(1.7,0) \rput[tl](1.7,-0.2){$a$} \psTick{90}(1,0) \rput[tl](1.1,0.3){\footnotesize$1$}
		\psTick{90}(-1.3,0) \rput[tc](-1.3,0.3){\footnotesize${-}\sqrt{a}$}\psline[style=Dash,linewidth=0.01,linecolor=gray](-1.7,1.7)(2,1.7)
		\psline[style=Dash,linewidth=0.01,linecolor=gray](-1.3,1.8)(-1.3,0)
		\psline[style=Dash,linewidth=0.01,linecolor=red](1,1)(1,0)\psline[style=Dash,linewidth=0.01,linecolor=red](-1,1)(-1,0)
		\psline[style=Dash,linewidth=0.01,linecolor=red](-1,1)(1,1)
		\psset{braceWidthInner=5pt,braceWidthOuter=5pt,braceWidth=0.1pt}
		\psbrace[ref=lC,linewidth=0.01](2,0)(2,1.7){$\dom{Y}_2$}\psbrace[ref=rC,rot=180,linewidth=0.01](-2,2.89)(-2,0){$\dom{Y}_1$}
		\psbrace[ref=t,rot=90,linewidth=0.01](-1.7,0)(0,0){$\dom{X}_1$}\psbrace[ref=t,rot=90,linewidth=0.01](0,0)(1.7,0){$\dom{X}_2$}
	\end{pspicture}
\caption{Piecewise strictly monotone function of Example 2}
\label{fig:sqlin}
\end{figure}
Consider an RV $X$ uniformly distributed on $[-a,a]$, $a\geq 1$, and a function $g(\cdot)$ defined as:
\begin{equation}
 g(x)=\begin{cases}
       x^2,&\text{if }x<0\\x,&\text{if }x\geq0
      \end{cases}.
\label{eq:squarelin}
\end{equation}
This function, depicted in Fig.~\ref{fig:sqlin}, is piecewise strictly monotone on $(-\infty,0)$ and $[0,\infty)$. We introduce the following partitioning:
\begin{IEEEeqnarray}{RCL}
 \dom{X}_1=[-a,0) &\rightarrow& \dom{Y}_1=(0,a^2]\\
 \dom{X}_2=[0,a]&\rightarrow& \dom{Y}_2=[0,a]
\end{IEEEeqnarray}
Since the function is not differentiable, we define $|g'(\cdot)|$ in a piecewise manner by the magnitude of the first derivatives of $g_l(\cdot)$:
\begin{equation}
 |g'(x)|=\begin{cases}
          2|x|,&\text{if }x<0\\1,&\text{if }x\geq0
         \end{cases}
\end{equation}
The two partial inverses on $\dom{Y}_2$ are given by 
\begin{IEEEeqnarray}{RCL}
 x_1&=&-\sqrt{y}=-\sqrt{g(x)}=
\begin{cases}
 x,&\text{if }x<0\\-\sqrt{x},&\text{if }x\geq0
\end{cases}\text{ and}\\
 x_2&=&y=g(x)=
\begin{cases}
 x^2,&\text{if }x<0\\x,&\text{if }x\geq0
\end{cases}
\end{IEEEeqnarray}
while on $\dom{Y}\backslash\dom{Y}_2=(a,a^2]$ only the root $x_1$ exists (i.e., this image is mapped to bijectively). Noticing that the information loss on the corresponding preimage $\dom{X}_b=[-a,-\sqrt{a})$ is zero, we can write for the conditional entropy:
\begin{IEEEeqnarray}{RCL}
\ent{X|Y} &=& 
\int_{\dom{X}_1\setminus\dom{X}_b} f_X(x) \log \left( \frac{\frac{f_X(x)}{2|x|}+f_X(x^2) }{\frac{f_X(x)}{2|x|} } \right)dx\notag\\
 &&{+}\:\int_{\dom{X}_2} f_X(x) \log \left( \frac{\frac{f_X(-\sqrt{x})}{2\sqrt{x}}+f_X(x) }{f_X(x) } \right)dx\notag\\
\end{IEEEeqnarray}
Since $f_X(x)=\frac{1}{2a}$ for all $x$, $x^2$, and $-\sqrt{x}$ in the designated integration ranges, we obtain
\begin{IEEEeqnarray}{RCL}
\ent{X|Y} &=& \int_{-\sqrt{a}}^{0} \frac{1}{2a} \log \left( 1+2|x| \right)dx\notag\\
 &&{+}\:\int_{0}^{a} \frac{1}{2a} \log \left( 1+\frac{1}{2\sqrt{x}} \right)dx\notag\\
&=& \frac{4a+4\sqrt{a}+1}{8a} \log (2\sqrt{a}+1)\notag\\
 &&{-}\:\frac{\log(2\sqrt{a})}{2}-\frac{1}{4\sqrt{a}\ln2}\notag
\end{IEEEeqnarray}
where $\ln$ is the natural logarithm. For $a=1$ this approximates to $\ent{X|Y}\approx 0.922$~bits. The information loss is slightly less than one bit, despite the fact that two equal probability masses collapse and the complete sign information is lost. This suggests that by observing the output part of the sign information can be retrieved. Looking at Fig.~\ref{fig:sqlin}, one can see that for the subdomain located on the negative real axis, i.e., for $\dom{X}_1$, more probability mass is mapped to smaller outputs than to higher outputs. Thus, for a small output value $y$ it is more likely that the input originated from $\dom{X}_1$ than from $\dom{X}_2$ (and vice-versa for large output values). Mathematically, this means that despite $p(w_1)=p(w_2)=0.5$, we have $p(w_1|y)\neq p(w_2|y)$, which according to Theorem~\ref{thm:equivToRoots} plays a central role in computing the conditional entropy $\ent{X|Y}$.

By evaluating the bounds from Theorem~\ref{thm:UpperBoundLoss} we obtain
\begin{equation}
 \ent{X|Y} \leq \frac{1+\sqrt{a}}{2\sqrt{a}} \leq \log\left(\frac{3\sqrt{a}+1}{2\sqrt{a}}\right)\leq 1
\end{equation}
which for $a=1$ all evaluate to 1 bit. The bounds are not tight as the conditions of Theorem~\ref{thm:UpperBoundLoss} are not met in this case.

\subsection{Example 3: Normal PDF, Third-order Polynomial}\label{ex:3rd}
\begin{figure}[t]
 \centering
	\begin{pspicture}[showgrid=false](-4,-1.25)(4,2)
		\psaxeslabels{->}(0,0)(-3,-1.25)(3,1.5){$x$}{$g(x)$}
		\psplot[style=Graph]{-2.6}{2.6}{x x mul x mul x -4 mul add 0.3 mul 2 div}
		\psTick{90}(-1.15,0) \rput[th](-1.15,-0.1){\footnotesize$-\frac{10}{\sqrt{3}}$}
		\psTick{90}(2.31,0) \rput[th](2.31,-0.1){\footnotesize$\frac{20}{\sqrt{3}}$}
		\psline[style=Dash,linewidth=0.01, linecolor=gray](-2.6,0.46)(2.6,0.46) 
	\end{pspicture}
\caption{Third-order polynomial of Example 3}
\label{fig:3rd}
\end{figure}
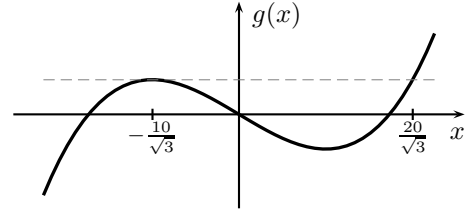
Finally, consider a Gaussian RV $X\sim\normdist{0}{\sigma^2}$ and the function
\begin{equation}
 g(x) = x^3-100x
\end{equation}
depicted in Fig.~\ref{fig:3rd}. An analytic computation of the information loss is prevented by the logarithm of a sum in~\eqref{eq:informationloss}. Still, we will show that with the help of Theorem~\ref{thm:UpperBoundLoss} at least a bound on the information loss can be computed.

Judging from the extrema of this function, three piecewise monotone intervals can be defined. 
Further, the domain which is mapped bijectively can be shown to be identical to
 $\dom{X}_b = \left(-\infty,-\frac{20}{\sqrt{3}}\right] \cup \left[\frac{20}{\sqrt{3}},\infty\right)$
and contains a probability mass of
\begin{equation}
 P_b = 2F_X\left(-\frac{20}{\sqrt{3}}\right) = 2Q\left(\frac{20}{\sqrt{3}\sigma}\right)
\end{equation}
where $Q(\cdot)$ is the $Q$-function. 
With this result and the fact that 
\begin{equation}
 P_b=\int_{\dom{X}_b} f_X(x)dx=\int_{\dom{Y}_b} f_Y(y)dy
\end{equation}
for a bijective mapping $g(\cdot)$ between $\dom{X}_b$ and $\dom{Y}_b$ we can upper bound the information loss by Theorem~\ref{thm:UpperBoundLoss}:
\begin{IEEEeqnarray}{RCL}
 \ent{X|Y} &\leq& \int_\dom{Y} f_Y(y)\log(|\indset{y}|)dy\\
&=&\int_{\dom{Y}\setminus\dom{Y}_b}  f_Y(y)\log3dy
= (1-P_b)\log 3
\end{IEEEeqnarray}
since $|\indset{y}|=1$ if $y\in\dom{Y}_b$ and $|\indset{y}|=3$ if $y\in\dom{Y}\setminus\dom{Y}_b$.
\begin{figure}[t]
 \centering
	\includegraphics[width=0.49\textwidth]{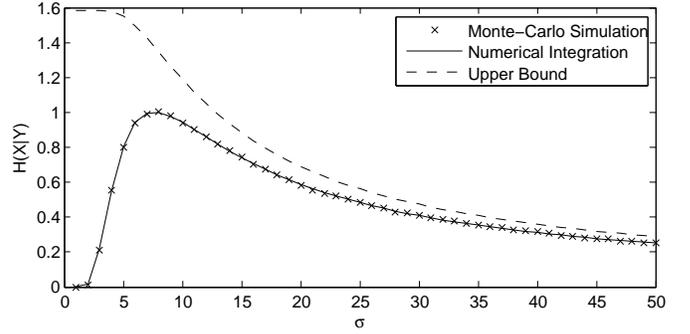}
\caption{Information loss of Example 3}
\label{fig:3rdorderSIM}
\end{figure}
In Fig.~\ref{fig:3rdorderSIM}, this bound is illustrated together with the results from numerical integration of~\eqref{eq:informationloss} and from Monte-Carlo simulations of the information loss.

\begin{figure*}[b]
 \normalsize
\hrulefill
 \begin{IEEEeqnarray}{RCL}
 \ent{X|Y} &=& \int f_{X}(x) \log\left(\frac{|g'(x)|f_{Y}(g(x))}{f_X(x)}\right)dx-\lim_{\hat{X}\rightarrow X}\iint f_{\hat{X}X}(\hat{x},x) \log\left(\sum_{i\in\indset{g(x)}}\frac{|g'(x)|f_{X|\hat{X}}(\hat{x},x_i)}{|g'(x_i)|f_{X|\hat{X}}(\hat{x},x)}\right)d\hat{x}dx\label{eq:long1}
\end{IEEEeqnarray}
\end{figure*}

\section{Conclusion}
We presented an analytic expression for the information loss induced by a static nonlinearity. It was shown that this information loss is strongly related to the non-injectivity of the system, i.e., to the fact that a particular output can result from multiple possible inputs. Conversely, given a certain output, the input to the system under consideration is uncertain only with respect to the roots of the equation describing the system. The information loss can be utilized, e.g., for estimating the reconstruction error for nonlinearly distorted signals.

Since the obtained expression involves the integral over the logarithm of a sum, bounds on the information loss were derived which can be computed with much less difficulty. In particular, it was shown that the information loss for a piecewise strictly monotone function is upper bounded by the logarithm of the number of subdomains, and that this bound is tight.

Generalizations of these results to rates of information loss and nonlinear systems with memory, as well as the extension to discrete random variables are the object of future work.

\section*{Acknowledgments}
The authors gratefully acknowledge discussions with Sebastian Tschiatschek concerning mathematical notation, and his comments improving the quality of this manuscript.

\appendix

We have yet to show that in the proof of Theorem~\ref{thm:InfoLoss} limit and integration can be exchanged. Since both $\mutinf{\hat{X};X}$ and $\mutinf{\hat{X};Y}$ are an expectation over the same joint PDF (cf.~\eqref{eq:mutinf3} and~\eqref{eq:mutinf1}), we can express their difference as a single integral. Splitting the logarithm of a product in a sum of logarithms, we obtain
\begin{IEEEeqnarray*}{RCL}
 \ent{X|Y} &=& \lim_{\hat{X}\rightarrow X} \left(\mutinf{\hat{X};X} - \mutinf{\hat{X};Y}\right)\\
&=&\lim_{\hat{X}\rightarrow X}\left[ \iint f_{\hat{X}X}(\hat{x},x) \log\left(\frac{f_{Y}(g(x))}{f_X(x)}\right)d\hat{x}dx\right.\\
&&{+}\: \left.\iint f_{\hat{X}X}(\hat{x},x) \log\left(\frac{f_{X|\hat{X}}(\hat{x},x)}{f_{Y|\hat{X}}(\hat{x},g(x))}\right)d\hat{x}dx\right].
\end{IEEEeqnarray*}
Note that the integration ranges have been omitted due to space limitations. The first double integral can be reduced to a single integral and taken out of the limit, since the logarithm does not depend on $\hat{x}$. For the second integral we first invoke the method of transformation (e.g.,~\cite{Papoulis_Probability}) to obtain
\begin{IEEEeqnarray*}{RCL}
 f_{Y|\hat{X}}(\hat{x},g(x)) &=& \sum_{i\in\indset{g(x)}}\frac{f_{X|\hat{X}}(\hat{x},x_i)}{|g'(x_i)|}\\
   &=& \frac{f_{X|\hat{X}}(\hat{x},x)}{|g'(x)|} \sum_{i\in\indset{g(x)}}\frac{|g'(x)|f_{X|\hat{X}}(\hat{x},x_i)}{|g'(x_i)|f_{X|\hat{X}}(\hat{x},x)}
\end{IEEEeqnarray*}
Using this result in the integral above and splitting again the logarithm, one obtains~\eqref{eq:long1}. Note that one element in the sum in the logarithm is identical to one, since the preimage of $g(x)$ contains $x$. All other elements of the sum are non-negative, so the integral is taken over a positive function.

Now let $\hat{X}$ approach $X$ in a special way: As depicted in Fig.~\ref{fig:sysmod}, $\hat{X}$ is a (non-uniformly) quantized version of $X$. We perform the limit as a sequence of refinements of the partitioning, such that eventually the step sizes approach zero. $X$ can thus be viewed as the sum of $\hat{X}$ and a signal-dependent noise term $N$ with support $\Delta_{\hat{x}}$. The subscript $\hat{x}$ indicates the dependence of the quantization step size on the quantization bin, i.e., the non-uniformity of the quantizer. Let $Q(\hat{X})$ denote the discrete values $\hat{X}$ can assume. We then obtain for the conditional PDF of $X$ given $\hat{X}$
\begin{equation}
 f_{X|\hat{X}}(\hat{x},x) = \frac{f_X(x)\chi_{\{|x-Q(\hat{x})|<\frac{\Delta_{\hat{x}}}{2}\}}}{F_X(Q(\hat{x})+\frac{\Delta_{\hat{x}}}{2})-F_X(Q(\hat{x})-\frac{\Delta_{\hat{x}}}{2})}
\end{equation}
where $\chi_{\{\cdot\}}$ is the indicator function, assuming one whenever the condition in the argument is fulfilled and zero otherwise.

The refinement of the partitioning suggests a converging sequence of functions as the argument of the integral in~\eqref{eq:long1}. At some point in the sequence we can assume that, for all $x$, the partitioning is so fine that no two elements of the preimage of $g(x)$ lie in the same quantization bin. This is trivially fulfilled if the end points of all intervals $\dom{X}_l$ (cf. Definition~\ref{def:function}) are quantization thresholds. In this case the sum in the logarithm in~\eqref{eq:long1} is identical to one, which renders the argument of the integral zero. Since any refinement of a partitioning does not change existing, but only adds new quantization thresholds, the integral remains zero.

While this analysis yields yet another proof for Theorem~\ref{thm:InfoLoss}, it also allows us to upper bound the argument of the second integral in~\eqref{eq:long1} by the zero function. Thus, invoking Lebesgue's dominated convergence theorem allows us to exchange limit and integration.

\bibliographystyle{IEEEtran}
\bibliography{IEEEabrv,/afs/spsc.tugraz.at/project/IT4SP/1_d/Papers/InformationProcessing.bib,%
/afs/spsc.tugraz.at/project/IT4SP/1_d/Papers/ProbabilityPapers.bib,%
/afs/spsc.tugraz.at/user/bgeiger/includes/textbooks.bib,%
/afs/spsc.tugraz.at/user/bgeiger/includes/myOwn.bib,%
/afs/spsc.tugraz.at/project/IT4SP/1_d/Papers/HMMRate.bib,%
/afs/spsc.tugraz.at/project/IT4SP/1_d/Papers/InformationWaves.bib,%
/afs/spsc.tugraz.at/project/IT4SP/1_d/Papers/ITBasics.bib,%
/afs/spsc.tugraz.at/project/IT4SP/1_d/Papers/ITAlgos.bib}

\end{document}